\documentclass[12pt]{article}
\usepackage{amsfonts}
\usepackage{amscd}
\usepackage{amsthm}
\usepackage{amsxtra}
\usepackage{syntonly}
\usepackage[all,2cell,dvips]{xy}
\usepackage{amssymb,amsmath}
\usepackage{fancyhdr}
\usepackage%[ps2pdf, hypertex, backref]
{hyperref}

\voffset = - 1.2cm
\newtheorem{thm}{Theorem}[section]
\newtheorem{lemma}{Lemma}[section]

\newtheorem{prop}{Proposition}[section]
\newtheorem{defn}{Definition}[section]

\theoremstyle{definition}

\newcommand{\w}{\omega}        % symplectic form
\newcommand{\ed}{\mathbf{d}}   % exterior diferential
\newcommand{\R}{\mathbb{R}}    % real line
\newcommand{\g}{\mathfrak{g}}% lie algebra g
\newcommand{\be}{\begin{equation}}% begin equation
\newcommand{\ee}{\end{equation}}% end equation
\newcommand{\bea}{\begin{eqnarray}}% begin equation array
\newcommand{\eea}{\end{eqnarray}}% end equation array
\newcommand{\restr}[1]{\vrule height3ex width.4pt depth1.4ex\lower1.4ex\hbox{\scriptsize $\,#1$}}
\newcommand{\rrestr}[1]{\vrule height2ex width.4pt depth0.9ex\lower0.9ex\hbox{\scriptsize $\,#1$}}
\newcommand{\LE}[1]{\mathcal{L}^{#1}E}

\begin{document}

\title{Gauge Equivalence and Conserved Quantities for Lagrangian Systems on Lie Algebroids}
\author{Jos\'e F. Cari\~nena$^{\dagger\,a}$ and Miguel
  Rodr\'{\i}guez-Olmos$^{\ddagger\,b}$\\ [2pt]
$^{\dagger)}$ Departamento de F\'{\i}sica Te\'orica, Facultad de Ciencias,\\
 Universidad de Zaragoza. 5009, Spain.\\
% jfc@unizar.es}
$^{\ddagger)}$ School of Mathematics, The University of Manchester,\\ Oxford Road, Manchester M13 9PL, UK.}
%miguel.rodriguez.olmos@manchester.ac.uk }

\maketitle

\begin{abstract}We develop a theory of gauge and dynamical equivalence for Lagrangian systems on Lie algebroids, also studying its relationship
with N\"other and non-N\"other conserved quantities.\end{abstract}

\begin{quote}
%---------------


%---------------


\end{quote}
{\vfill}
\footnoterule
{\noindent\small
$^{a)}${\it E-mail address:} {jfc@unizar.es}  \\
$^{b)}${\it E-mail address:} {miguel.rodriguez.olmos@manchester.ac.uk}

\section{Introduction}

\qquad It has been proved during the last years that Lie algebroids \cite{HiggMac, Mac}
provide a very
general framework for dealing with different problems in Mechanics
 \cite{CM,Wein}
and control theory \cite{Cort04},
including reduction of mechanical systems with symmetries \cite{CarNunSan05,CarNunSan07a}. The concept of
Lie algebroid is a generalisation of both a Lie algebra and a
tangent bundle structure, these being the simplest examples of Lie
algebroids. Moreover, the Lie algebroid structure is well adapted to
variational calculus for constrained systems \cite{GG}
and the
geometric treatment of the concept of
quasi-coordinates finds its natural place in
Lie algebroid framework \cite{CarNunSan07b}. The usefulness of the Lie
algebroid approach for dealing with non-holonomic constrained system
is also beyond doubt \cite{Cort05}.

The geometric theory of Lagrangian formulation in Lie algebroids started in
\cite{Wein}
and developed in \cite{Mart,LMM} without using the Legendre transformation (see also
\cite{GGU})
gives rise to the problem of existence and uniqueness of a Lagrangian inducing a previously fixed dynamics and it is the uniqueness property which is
going to be analysed in this paper.

The organization of the paper is as follows: in Section 2 we give a concise
survey of the theory of Lagrangian mechanics on Lie algebroids, following
\cite{Wein, Mart, LMM, CarMart}. In Section 3 we introduce several notions of equivalence of Lagrangian systems on Lie algebroids, obtaining a   relationship among them in Theorem \ref{them gauge geo}. This generalises the results obtained in \cite{CI} for standard Lagrangian systems. Section 4 contains the main results of the paper: Theorem \ref{thm one param gauge} gives sufficient conditions under which we can obtain a one parameter family of Lagrangian functions gauge equivalent to a given one. This is related in Theorem \ref{noether} to the existence of N\"other conserved quantities for the Lagrangian dynamics on Lie algebroids. Finally, in Section 5 we study conserved quantities which are not of N\"other type in the case that the Lagrangian dynamics is a locally Hamiltonian symmetry of two different 2-forms on the Lie algebroid.

\section{Lie Algebroids and Lagrangian Mechanics}
\qquad This section collects some basic results about the geometry of Lie algebroids and the formulation
 of Lagrangian mechanics on them. All the results here are standard and easy to find in the literature, so this section remains purely expository.

\subsection{Lie algebroids} Recall (\cite{CM,Wein}) that a Lie algebroid is a
vector bundle $\tau^E:E\rightarrow M$, together with a Lie algebra
structure $[\cdot,\cdot]_E$ in the space of sections
$\Gamma(\tau^E)$ and a vector bundle morphism (anchor map)
$\rho^E:E\rightarrow TM$ satisfying the compatibility condition
$$[X,fY]_E=\rho^E(X)fY+f[X,Y]_E,\quad f\in C^\infty(M).$$
Extremal cases of Lie algebroids include the tangent bundle of $M$, $\tau_M:TM\rightarrow M$ with the usual Lie
algebra structure on $\Gamma(\tau_{TM})=\mathfrak{X}(M)$ given by the Lie bracket and
$\rho^{TM}=\text{id}$, or a Lie algebra $\mathfrak{g}$, thought as a
vector bundle over a point, for which the anchor is trivial and
$\Gamma({\tau^\g})=\g$ with its own Lie algebra data. Many other
examples can be obtained from classic geometric structures, such as
foliations and Poisson or Dirac geometry.

Morally, in the formulation of Lagrangian mechanics on Lie
algebroids the algebroid $E$ replaces the tangent bundle to a
manifold $M$ as the space of positions and velocities for the dynamics. In the same way,
sections of the exterior algebra $\Omega^\bullet(E)$ of the dual
bundle $\tau^{E^*}:E^*\rightarrow M$ play the role of
`generalised' differential forms. A differential calculus on a Lie
algebroid can be built on a graded derivation of degree 1,
$\ed^E:\Omega^k(E)\rightarrow\Omega^{k+1}(E)$, which takes the place
of the usual exterior derivative in this context and is defined by
\begin{eqnarray*} \ed^E\theta(X_1,\ldots,X_{k+1}) & = & \sum_i
(-1)^{i+1}\rho^E(X_i) \theta(X_1,
\ldots,\widehat X_i,\ldots,X_{k+1})\\
& + & \sum_{i<j}(-1)^{j+i}\theta([X_i,X_j]_E,X_1,\ldots,\widehat
X_i,\ldots, \widehat X_j,\ldots,X_{k+1})\ , \end{eqnarray*} where
$X_1,\ldots,X_{k+1}\in\Gamma(\tau^E)$ and $\theta\in \Omega^k(E)$.

This exterior derivative reduces to the usual exterior derivative when
$E=TM$. Its associated Lie
derivative $\ed^E_X$ along $X\in \Gamma(\tau^E)$ is defined by
\begin{eqnarray}\nonumber
\ed^E_Xf & = & \rho^E(X)f,\\
\nonumber \ed^E_X Y & = & [X,Y]_E,\\
\label{magic formula} \ed^E_X \alpha & = & (\iota_X\circ\ed^E+\ed^E\circ\iota_X)\alpha,
\end{eqnarray}
for $X,Y\in\Gamma (\tau^E),\,f\in
C^\infty(M),\,\alpha\in\Omega^\bullet(E)$. The following properties
generalise some well-known properties of the usual Lie derivative
(see e.g. \cite{Mart,CarMart} for a proof).
\begin{eqnarray*}
\ed^E\circ\ed^E & = & 0,\\
\left[\ed^E,\ed^E_X\right] & = & 0,\\
\ed^E_{[X,Y]_E} & = & [\ed^E_X,\ed^E_Y],\\
\iota_{[X,Y]_E} & = & [\ed^E_X,\iota_Y],\\
\ed^E_X(\alpha\wedge\beta) & = & \ed^E_X\alpha\wedge\beta
+\alpha\wedge\ed^E_X\beta.
\end{eqnarray*}

\subsection{The prolongation of a Lie algebroid} Let $M'$ be a smooth
manifold and let $f:M'\rightarrow M$ a fibration. For each point $x'\in M'$ let
$\mathcal{T}_{x'}^EM'$ be the linear space
$$
\mathcal{T}_{x'}^EM'=\{(a,v)\in E_{x}\times T_{x'}M'\mid \rho^E(a)=T_{x'}f(v)\},
$$
where $f(x')=x$ and $Tf$ denotes the tangent map to $f$, $Tf:TM'\to TM$. The
set
$\mathcal{T}^EM'=\bigcup_{x'\in M'}\mathcal{T}_{x'}^EM'$ is endowed with a natural vector
bundle structure over $M'$; the vector bundle projection  $\tau_{M'}^E:\mathcal{T}^EM'\to M'$
is $\tau_{M'}^E(a,v)=\tau_{M'}(v)$. Note that in particular we have $\tau^E(a)=(f\circ \tau_{M'})(v)$, with
$\tau_{M'}:TM'\rightarrow M'$
being  the tangent bundle projection.

 Moreover, such a vector bundle can be
endowed with a Lie algebroid structure, the anchor map being the projection onto
the second factor and the bracket on the linear space of sections is the only
one that for two sections of the form $Y_i(x')=(x',\sigma_i(f(x')),U_i(x'))$,
$i=1,2$, with $\sigma_i\in \Gamma(\tau^E)
,U_i\in\mathfrak{X}(M)$, satisfies
$$\lbrack\!\lbrack Y_1,Y_2\rbrack\!\rbrack(x')=(x',[\sigma_1,\sigma_2]_E(f(x')), [U_1,U_2](x'))\,.
$$

This Lie algebroid is called the \emph{prolongation of $E$ along $f$} and more
details about it can be found in \cite{LMM,{MMS}}.

 The interesting case for the Lagrangian
formulation of Mechanics on Lie algebroids is when $M'=E$ and $f=\tau^E$.
In this case we will denote simply by $\LE{}$ the corresponding prolongation of
$E$ along $\tau^E$
 and it is possible to show that it can also be realised as the total space of the pullback bundle $({\rho^E})^*T\tau^E$.
This is the particular case studied in
 \cite{Mart} (see also \cite{LMM}). Then $\LE{}$ has a Lie algebroid structure
 that we next describe:

\begin{enumerate}

\item $\LE{}$ is a vector bundle over $E$ with projection $\tau^{\LE{}}:\LE{}\rightarrow E$ given by $\tau^{\LE{}}(a,v)=\tau_E (v)$.

\item It can be proved that every section $X\in\Gamma (\tau^{\LE{}})$ can be written as $$X(a)=(f_X(a)X_1(\tau^E(a)),X_2(a)),$$
where $X_1\in \Gamma (\tau^E),\,X_2\in\mathfrak{X}(E)$ and $f_X\in C^\infty(E)$ satisfying the condition
$$T_a\tau^E(X_2(a))=f_X(a)\rho^E(X_1(\tau^E(a))),\quad\forall a\in E.$$

\item The Lie bracket in $\Gamma(\tau^{\LE{}})$, denoted $[\cdot,\cdot]_{\LE{}}$, is given by
$$[X,Y]_{\LE{}}=\left(f_Xf_Y([X_1,Y_1]_E\circ\tau^E)+(X_2f_Y)(Y_1\circ\tau^E)-(Y_2f_X)(X_1\circ\tau^E),[X_2,Y_2]\right).$$

\item The anchor $\rho^{\LE{}}:\LE{}\rightarrow TE$ is given by
$\rho^{\LE{}}(X)=X_2$.
\end{enumerate}

\subsection{Complete and vertical lifts}
We start by defining lifts to $E$ of functions in $M$. If $f\in
C^\infty(M)$, its \emph{complete and vertical lifts} are the
functions $f^c,f^v\in C^\infty(E)$ defined as
$$f^c(a)=\rho^E(a)f,\quad f^v=f\circ\tau^E,\quad \forall\, a\in E$$

Let $X$ be a section of $\tau^E$. We can also define its complete
and vertical lifts $X^\text{vert},X^\text{comp}\in\mathfrak{X}(E)$
as follows: First, let $\alpha\in\Omega^1(E)$, then
$\widehat{\alpha}\in C^\infty (E)$ is defined as
$$\widehat{\alpha}(a)=\langle\alpha(\tau^E(a)),a\rangle.$$

With this notation we have:

\begin{enumerate}
\item The vertical lift of $X$ is given by
$$X^\text{vert}(a)=\frac{d}{dt}(a+tX(\tau^E(a)))\rrestr{t=0}\,.$$

\item The complete lift of $X$ is the unique $\tau^E$-projectable vector field on $E$ satisfying
$${\tau^E}_*(X^\text{comp})=\rho^E(X)\quad \text{and} \quad X^\text{comp}(\widehat{\alpha})=\widehat{\ed^E_X\alpha}.$$
\end{enumerate}

Finally, we can also lift sections of $\tau^E:E\to M$ to sections of the
prolongation $\tau^{\LE{}}:\LE{}\to E$. For any $X\in\Gamma (\tau^E)$ the complete and
vertical lifts, $X^c,X^v\in\Gamma(\tau^{\LE{}})$ are:
$$X^v(a)=(0,X^\text{vert}(a)),\quad
X^c(a)=(X(\tau^E(a)),X^\text{comp}(a)).
$$
The set of complete and vertical lifts generates
$\Gamma(\tau^{\LE{}})$ as a $C^\infty(E)$-module. Therefore,  the Lie algebroid structure of
$\LE{}$ is characterised by
\begin{equation}\label{props prolongation}\begin{array}{lll}
[X^c,Y^c]_{\LE{}}=([X,Y]_E)^c, &  [X^c,Y^v]_{\LE{}}=([X,Y]_E)^v, & [X^v,Y^v]_{\LE{}}=0\vspace{2mm}\\
\rho^{\LE{}}(X^c)(f^c)=(\rho^E(X)(f))^c, & \rho^{\LE{}}(X^c)(f^v)=(\rho^E(X)(f))^v, & \vspace{2mm}\\
\rho^{\LE{}}(X^v)(f^c)=(\rho^E(X)(f))^v, &
\rho^{\LE{}}(X^v)(f^v)=0,\end{array}
\end{equation}
for $X,Y\in\Gamma(\tau^E)$ and $f\in C^\infty(M)$.

\subsection{The Euler section and the vertical endomorphism} The \emph{Euler section} is the section of $\tau^{\LE{}}$
defined by $$\Delta(a)=(0,a^v_a),\quad \forall a\in E,$$ where $0$ is the zero
element of ${\tau^E}^{-1}(\tau^E(a))$ and
$$a_a^v=\frac{d}{dt}(a+ta)\rrestr{t=0}\in T_aE\,.$$

The following property of such a section be useful later on (see
\cite{Mart} for a proof).
\begin{lemma}\label{liouvillecomplete}
Let be $X\in\Gamma (\tau^E)$. Then
$$[\Delta,X^v]_{\LE{}}=-X^v\quad\text{and}\quad [\Delta,X^c]_{\LE{}}=0.$$
\end{lemma}
The \emph{vertical endomorphism} is the vector bundle automorphism
$S:\LE{}\rightarrow\LE{}$ such that for each section $X\in\Gamma(\tau^E)$,
\begin{equation}S(X^v)=0\quad\text{and}\quad S(X^c)=X^v.
\end{equation}
Note that all the concepts so far introduced restrict in the case
$E=TM$ to the usual constructions for the double tangent bundle
$TTM$. In the same vein one can define a second order differential
equation (SODE) on $E$ as a section $X\in\Gamma(\tau^{\LE{}})$ satisfying
\begin{equation}\label{SODE}S(X)=\Delta.\end{equation}

\subsection{Lagrangian mechanics} With the ingredients previously introduced, we can formulate the problem of Lagrangian
mechanics on Lie algebroids. For that, we will follow  closely \cite{Mart}
and \cite{LMM}.  Let $L\in C^\infty(E)$ be a Lagrangian. The
\emph{Poincar\'e 1- and 2-forms} relative to $L$ are defined by
\begin{eqnarray}
\label{theta} \Theta_L & = & \ed^{\LE{}}L\circ S,\\
\label{omega} \omega_L & = & -\ed^{\LE{}}\Theta_L.
\end{eqnarray}
The \emph{energy function} $E_L\in C^\infty(E)$ is
$$E_L=\rho^{\LE{}}(\Delta)(L)-L.$$

Then, the dynamics associated to
$(E,M,\tau^E,[\cdot,\cdot]_E,\rho^E,L)$, or to $L$ in short, are the projection to $M$ by
$\tau^{E}$ of the integral curves of the vector field
$\rho^{\LE{}}(Z_L)\in\mathfrak{X}(E)$, where
$Z_L\in\Gamma(\tau^{\LE{}})$ is a solution of the equation
\begin{equation}\label{EL}
\iota_{Z_L}\omega_L=\ed^{\LE{}}E_L.\end{equation}
We will refer to $Z_L$ as the dynamics associated to $L$ as well.
The form $\w_L$ is
called \emph{non-degenerate} if $\w_L(X,Y)=0$ for every $Y\in\Gamma
(\tau^{\LE{}})$ implies $X=0$. In case $\omega_L$ is non-degenerate
$L$ is called \emph{regular}, and there is a unique solution $Z_L$
to \eqref{EL} which is also a SODE in the sense of \eqref{SODE}. In
this paper we will be concerned only with regular Lagrangians.

\subsection{Local expressions} We provide now concrete
local expressions for the objects defined so far. Suppose that
$\text{dim}\,M=m,\,\text{rank}\, E=p$ and let $U\subset E$ be a
trivializing open neighborhood coordinatised by $\{x^i,y^\alpha\}$
where $\{x^i\},i=1,\ldots,m$, are local coordinates on $M$ and
$\{y^\alpha\},\alpha=1,\ldots,p$, are linear coordinates on the
typical fiber of $E$ relative to a local basis of
$\Gamma(\tau^E)\rrestr{\tau^E(U)}$ given by
$\{e_\alpha\},\,\alpha=1,\ldots,p$. The local structure of $E$ is
encoded in the structure functions $C_{\alpha\beta}\,^\gamma,\, \rho^i_\alpha\in
C^\infty(\tau^E(U))$,  defined by
\begin{equation}\label{localstructure}
[e_\alpha,e_\beta]_E=C_{\alpha\beta}^\gamma\,e_\gamma,\quad
\rho^E(e_\alpha)=\rho^i_\alpha\frac{\partial}{\partial
x^i},\end{equation} where $i=1,\ldots,m$ and
$\alpha,\beta,\gamma=1,\ldots,p$.

It is shown in \cite{LMM} that a local basis of
$\Gamma(\tau^{\LE{}})\rrestr{U}$ is given by
$\{\widetilde{T}_\alpha,\widetilde{V}_\beta\},\,\alpha,\beta=1,\ldots,p$,
where \begin{equation}\label{localbasis} \widetilde{T}_\alpha (a) =
\left(e_\alpha (\tau^E(a)),\rho^i_\alpha\frac{\partial}{\partial
x^i}\restr{a}\right),\quad \widetilde{V}_\alpha (a)  =
\left(0,\frac{\partial}{\partial y^\alpha}\restr{a}\right),\quad
a\in U.
\end{equation}
We provide for later use the local expression of the two-form $\w_L$ in this coordinates: With respect to the dual basis $\{\widetilde{T}^\alpha,\widetilde{V}^\beta\}$ we have that, on $U$,
\begin{equation}\label{localsymplecticform}
\w_L=\frac{\partial^2L}{\partial y^\alpha\partial y^\beta}\widetilde{T}^\alpha\wedge\widetilde{V}^\beta+\left(\frac 12\frac{\partial L}{\partial y^\gamma}C^\gamma_{\alpha\beta}-\rho^i_\alpha\frac{\partial^2 L}{\partial x^i\partial y^\beta}\right)\widetilde{T}^\alpha\wedge\widetilde{T}^\beta.\end{equation}
Since we are assuming that $L$ is non-degenerate, the matrix
with elements $\left[\frac{\partial^2L}{\partial y^\alpha\partial
    y^\beta}\right]$, where ${\alpha,\beta\in(1,\ldots,p)}$, is invertible, and there is a unique solution $Z_L$ for \eqref{EL} with associated dynamics characterised by
\begin{equation}\label{localEL}
\frac{dx^i}{dt}=\rho^i_\alpha y^\alpha,\quad \frac {d}{dt}\left(\frac{\partial
    L}{\partial y^\alpha}  \right)=\rho^i_\alpha\frac{\partial L}{\partial
  x^i}-C^\gamma_{\alpha\beta}y^\beta\frac{\partial L}{\partial y^\gamma},\qquad
i=1,\ldots,m\,.
\end{equation}

\section{Gauge Equivalent Lagrangians}
In this section we introduce different notions of equivalence among
Lagrangians defined on $E$.

\subsection{Basic forms} In order to build different notions
of equivalence of Lagrangians, we will need the concept of basic and
semi-basic 1-forms, as well as several of their properties.

\begin{defn}
A 1-form $\theta\in \Omega^1(\LE{})$ is called semi-basic  if
$\theta(X^v)=0$ for every $X\in\Gamma (\tau^E).$ A 1-form $\theta\in
\Omega^1 (\LE{})$ is called basic if it is semi-basic and in
addition there is a unique form $\theta'\in \Omega^1 (E)$ such that
$\theta (X^c)=\theta' (X)\circ\tau^E$ for every $X\in \Gamma
(\tau^E)$. This defines a  linear bijection $\gamma$ between $\Omega^1(E)$
and $\Omega^1_\text{{\rm bas}}(\LE{})$, the space of basic 1-forms on
$\LE{}$, by $\gamma^{-1}(\theta)=\theta'$.
\end{defn}

The following two propositions collect some important properties of basic and semi-basic forms.
\begin{prop} \label{propbasic1} A closed semi-basic 1-form $\beta\in\Omega^1 (\LE{})$ is basic. In that case $\beta=\gamma (\beta')$ with $\ed^E\beta'=0$.
\end{prop}
\begin{proof}
Let be $X,Y\in\Gamma(E)$. Since $\beta$ is closed,
\begin{eqnarray*}
0 & = &
\ed^{\LE{}}\beta(X^c,Y^v)=\rho^{\LE{}}(X^c)(\beta(Y^v))-\rho^{\LE{}}(Y^v)(\beta(X^c))-
\beta\left([X^c,Y^v]_{\LE{}}\right)\\ & = &
-Y^\text{vert}(\beta(X^c)),
\end{eqnarray*}
since $\beta$ annihilates vertical sections and $[X^c,Y^v]_{\LE{}}=\left([X,Y]_E\right)^v$.
The above result implies that $\beta(X^c)={\tau^E}^*f$ for some $f\in C^\infty(M)$, and therefore
there exists $\beta'\in\Omega^1(E)$ such that $\gamma (\beta')=\beta$. That
$\beta'$ is closed follows again from the closeness of $\beta$, since
\begin{eqnarray*}
0 & = & \ed^{\LE{}}\beta(X^c,Y^c)=\rho^{\LE{}}(X^c)(\beta(Y^c))-\rho^{\LE{}}(Y^c)(\beta(X^c))-\beta\left([X^c,Y^c]_{\LE{}}\right)\\
 & = & \rho^{\LE{}}(X^c)(\beta'(Y)\circ\tau^E)-\rho^{\LE{}}(Y)(\beta'(X)\circ\tau^E)-\beta\left(\left([X,Y]_E\right)^c \right)\\
  & = & \rho^{\LE{}}(X^c)(\beta'(Y)\circ\tau^E)-\rho^{\LE{}}(Y)(\beta'(X)\circ\tau^E)-\beta'\left([X,Y]_E\right)\circ\tau^E \\
& = & \rho^{\LE{}}(X^c)(\beta'(Y))^v-\rho^{\LE{}}(Y)(\beta'(X))^v-(\beta'\left([X,Y]_E\right))^v \\
& = & (\rho^{E}(X)(\beta'(Y)))^v-(\rho^{E}(Y)(\beta'(X)))^v-(\beta'\left([X,Y]_E\right))^v \\
& = & (\rho^{E}(X)(\beta'(Y))-\rho^{E}(Y)(\beta'(X))-\beta'\left([X,Y]_E\right))^v \\
& = & (\rho^{E}(X)(\beta'(Y))-\rho^{E}(Y)(\beta'(X))-\beta'\left([X,Y]_E\right))\circ\tau^E \\
& = & \left(\ed^E\beta'(X,Y)\right)\circ\tau^E,
\end{eqnarray*}
and hence from the surjectivity of $\tau^E$, this is equivalent to
$\ed^E\beta'=0$.
\end{proof}
\begin{prop}
A 1-form $\beta\in\Omega^1(\LE{})$ is basic if and only if is semi-basic and $\ed^{\LE{}}_{X^v}\beta=0$ for every $X\in\Gamma(\tau^E)$.
\end{prop}
\begin{proof}
Let be $Y\in\Gamma (\tau^E)$ and suppose that $\beta$ is basic. Then
obviously
$$(\ed^{\LE{}}_{X^v}\beta)(Y^v)=0.$$
 Now
$(\ed^{\LE{}}_{X^v}\beta)(Y^c)=\rho^{\LE{}}(X^v)(\beta(Y^c))$. Since
by hypothesis $\beta(Y^c)=(\beta'(Y))\circ\tau^E=(\beta'(Y))^v$, then
$$\rho^{\LE{}}(X^v)(\beta(Y^c))=\rho^{\LE{}}(X^v)(\beta'(Y))^v=0,$$ and therefore
$\ed^{\LE{}}_{X^v}\beta=0$.

Suppose now that $\beta$ is semi-basic, and that
$(\ed^{\LE{}}_{X^v}\beta)(Y^c)=\rho^{\LE{}}(X^v)(\beta(Y^c))=0$,
therefore $\beta(X^c)$ is the pullback by $\tau^E$ of a function in
$M$ and then $\beta$ is basic.
\end{proof}

\subsection{Equivalence of Lagrangians} One can define
different notions of equivalence classes of Lagrangian functions on Lie algebroids as the sets of Lagrangians
that produce the same Poincar\'e 2-section or the same dynamical
section $Z_L$. The stronger notion of gauge
equivalence will be related in the next section to the existence of
conserved quantities for the associated Lagrangian dynamics.
\begin{defn}\label{def equivalence}
Let $L,L'\in C^\infty (E)$ be two (regular) Lagrangians. We will say that $L$ is geometrically equivalent to $L'$ if $\w_L=\w_{L'}$. We will say that
$L$ is equivalent to $L'$  if $Z_L=Z_{L'}$.
\end{defn}

Let us note that since $\w_{L_1+L_2}=\w_{L_1}+\w_{L_2}$, we have that  $L$ and $L'$ are geometrically equivalent if and only if $L'=L+L_0$ with
$\w_{L_0}=0$, and hence $L_0$ must be singular. The next result characterises this class of singular Lagrangians with trivial Poincar\'e 2-forms.

\begin{thm}
A Lagrangian $L_0\in C^\infty(E)$ satisfies $\w_{L_0}=0$ if and only if
$$L_0=\widehat{\alpha}+V\circ\tau^E,$$
 where $\alpha\in\Omega^1(E)$ is
closed and $V\in C^\infty (M)$.
\end{thm}
\begin{proof}First note that for any $X\in\Gamma(E)$ we have $\Theta_{V\circ\tau^E}(X^v)=0$ since vertical sections are in the kernel of the vertical
 endomorphism $S$. Also, $$\Theta_{V\circ\tau^E}(X^c)=\rho^{\LE{}}(S(X^c))(V\circ\tau^E)=\rho^{\LE{}}(X^v)(V\circ\tau^E)=0.$$
 Then
$\Theta_{V\circ\tau^E}=0$ and so $\w_{V\circ\tau^E}=0$.

For the contribution of $\alpha$ we have $\Theta_{\widehat{\alpha}}(X^v)=\ed^{\LE{}}\widehat{\alpha}(S(X^v))=0$ and
\begin{eqnarray*}\Theta_{\widehat{\alpha}}(X^c)(a) & = & \ed^{\LE{}}\widehat{\alpha}(S(X^c))(a)=\ed^{\LE{}}\widehat{\alpha}(X^v)=
X^\text{vert}(\widehat{\alpha})(a)\\ & = &
\frac{d}{dt}(\widehat{\alpha}(a+tX(\tau^E(a)))\restr{t=0}=\frac{d}{dt}\langle\alpha(\tau^E(a)),
a+tX(\tau^E(a))\rangle\restr{t=0}\\ & = & \alpha(X)(\tau^E(a)),\end{eqnarray*}
and therefore $\Theta_{\widehat{\alpha}}(X^c)=\alpha(X)\circ
\tau^E$. It is now straightforward to obtain
\begin{eqnarray*}\ed^{\LE{}}\Theta_{\widehat{\alpha}}(X^v,Y^v) & = & 0,\\
\ed^{\LE{}}\Theta_{\widehat{\alpha}}(X^c,Y^v) & = & -\rho^{\LE{}}(Y^v)(\alpha(X)\circ\tau^E)=0,\\
\ed^{\LE{}}\Theta_{\widehat{\alpha}}(X^c,Y^c) & = & \ed^E\alpha
(X,Y)\circ\tau^E=0,\end{eqnarray*} where the last term
vanishes since $\alpha$ is closed. Therefore
$\w_{\widehat{\alpha}}=\ed^{\LE{}}\Theta_{\widehat{\alpha}}=0$.

Conversely, since by its very definition \eqref{theta} we have that
$\Theta_{L_0}$ is semi-basic, if $\w_{L_0}$ vanishes then by
Proposition \ref{propbasic1}  $\Theta_{L_0}$ is basic and
$\Theta_{L_0}(X^c)=\alpha(X)\circ\tau^E$ for some closed 1-form
$\alpha\in\Omega^1(E)$. This is the same as
$$\ed^{\LE{}}L_0(S(X^c))=\rho^{\LE{}}(X^v)L_0=\alpha(X)\circ\tau^E,$$
which implies that $L_0=\widehat{\alpha}+V\circ\tau^E$, for any
$V\in C^\infty(M)$.
\end{proof}
\begin{defn} We  say that two regular Lagrangians $L,L'\in C^\infty(E)$ are gauge
equivalent if there exist $\alpha\in\Omega^1(E)$ and $V\in
C^\infty(M)$ such that $\ed^E\alpha=0,\ed^EV=0$ and
$L'=L+\widehat{\alpha}+V\circ\tau^E$.
\end{defn}
Obviously if $L$ and $L'$ are gauge equivalent, then they are
geometrically equivalent. Note that the condition $\ed^EV=0$ does
not necessarily implies that $V$ is locally constant, since it
suffices that $\ed V$ annihilates the image of the anchor $\rho^E$.
The next theorem gives the relationship between the three different
notions of equivalence of Lagrangians. We start with a necessary
technical lemma.

\begin{lemma}\label{lemadelta}
Let be $\alpha\in\Omega^1(E)$ and $\Delta$ be the Liouville section. Then
$\ed_\Delta^{\LE{}}\widehat{\alpha}=\widehat{\alpha}.$
\end{lemma}
\begin{proof}Let $a\in E$, then \begin{eqnarray*}
\ed_\Delta^{\LE{}}\widehat{\alpha} (a)& = & \rho^{\LE{}}(\Delta)\widehat{\alpha}(a) =\frac{d}{dt}\widehat{\alpha}(a+ta)\restr{t=0}\\
& = & \frac{d}{dt}\langle\alpha(\tau^E(a)),a +ta\rangle=\langle\alpha(\tau^E(a)),a\rangle\restr{t=0} = \widehat{\alpha}(a).
\end{eqnarray*}

\end{proof}

\begin{thm}\label{them gauge geo}
Two (regular) Lagrangians $L$ and $L'$ are gauge equivalent if and only if they are equivalent and geometrically equivalent.
\end{thm}
\begin{proof}
If $L$ and $L'$ are gauge equivalent, then $L'=L+\widehat{\alpha}+V\circ\tau^E$ with $\alpha$ and $V$ closed by $\ed^E$. Then $\w_L=\w_{L'}$.
Notice also that since $\rho^{\LE{}}(\Delta)$ is a vector field tangent to the $\tau^E$-fibers of $E$, then $\ed^{\LE{}}_\Delta (V\circ\tau^E)=0$.
Then, for the energies of both Lagrangians,
$$E_{L'}=\ed^{\LE{}}_\Delta L'-L'=\ed^{\LE{}}_\Delta L+\ed^{\LE{}}_\Delta\widehat{\alpha}-L-\widehat{\alpha}-V\circ\tau^E=E_L-V\circ\tau^E,$$
where we used Lemma \ref{lemadelta}. Now, since by hypothesis $V$ is
closed, $\ed^{\LE{}}(V\circ\tau^E)=0$ and then
$\ed^{\LE{}}E_L=\ed^{\LE{}}E_{L'}$, which implies $Z_{L'}=Z_{L}$.

Conversely, if $L$ and $L'$ are geometrically equivalent, then
$L'=L+\widehat{\alpha}+V\circ\tau^E$ with $\ed^E\alpha=0$. Then,
proceeding as before, $E_{L'}=E_L-V\circ\tau^E$. Since
$\w_L=\w_{L'}$ and both Lagrangians are regular, the equivalence of $L$
and $L'$ forces $\ed^{\LE{}}E_{L'}= \ed^{\LE{}}E_{L}$, implying
$\ed^{\LE{}}(V\circ\tau^E)=0$, which is equivalent to $\ed^EV=0$.
\end{proof}

\section{Gauge Equivalence and N\"other's Theorem}
In this section we will study sections in $\Gamma(\tau^{\LE{}})$ that generate
one-parameter families of gauge equivalent Lagrangian functions. In Theorem
\ref{noether} we will show that  one can associate a N\"other conserved
quantity of the dynamics to each such a family.

\subsection{Admissible sections} We will start by studying admissible
sections, that is, those sections in $\Gamma(\tau^{\LE{}})$
preserving the different objects of the dynamical equation \eqref{EL}.
\begin{defn}
Let $X\in\Gamma(\tau^{\LE{}})$. We  say that $X$ is an admissible section for
$\w_L$
(respectively for $\Theta_L,\,E_L$)
if $\ed^{\LE{}}_X\w_L=\w_{\ed^{\LE{}}_XL}$  for all $L\in C^\infty(E)$ (respectively
$\ed^{\LE{}}_X\Theta_L=\Theta_{\ed^{\LE{}}_XL},\,\ed^{\LE{}}_XE_L=E_{\ed^{\LE{}}_XL}$).
\end{defn}
We start by studying the transformation of the elements in \eqref{EL} under the action of arbitrary sections in $\Gamma(\tau^{\LE{}})$.
\begin{lemma}\label{lemachanges}
Let $X,Y\in\Gamma(\tau^{\LE{}})$. Then,
\begin{eqnarray*}
(\ed^{\LE{}}_X\Theta_L)(Y) & = & \Theta_{\ed^{\LE{}}_XL}(Y)+\ed^{\LE{}}L\left([X,S(Y)]_{\LE{}}-S([X,Y]_{\LE{}})\right)\\
\ed^{\LE{}}_X\w_L & = & -\ed^{\LE{}}\ed^{\LE{}}_X\Theta_L\\
\ed^{\LE{}}_XE_L & = & E_{\ed^{\LE{}}_XL}+\ed^{\LE{}}_{[X,\Delta]_{\LE{}}}L
\end{eqnarray*}
\end{lemma}
\begin{proof}
For the first property, \begin{eqnarray*}
(\ed^{\LE{}}_X\Theta_L)(Y) & = & \ed^{\LE{}}_X(\Theta_L(Y))-\Theta_L([X,Y]_{\LE{}})\\
& = & \ed^{\LE{}}_X ( \ed^{\LE{}}L(S(Y))-\ed^{\LE{}}L(S([X,Y]_{\LE{}}))\\
 & = & (\ed^{\LE{}}_X\ed^{\LE{}}L)(S(Y))+\ed^{\LE{}}L(\ed^{\LE{}}_XS(Y))-\ed^{\LE{}}L(S([X,Y]_{\LE{}}))\\
& = & (\ed^{\LE{}}\ed^{\LE{}}_XL)(S(Y))+\ed^{\LE{}}L([X,S(Y)]_{\LE{}}-S([X,Y]_{\LE{}}))\\
& = & \Theta_{\ed^{\LE{}}_XL}(Y)+\ed^{\LE{}}L([X,S(Y)]_{\LE{}}-S([X,Y]_{\LE{}})).
\end{eqnarray*}
The second property is immediate since $\w_L=-\ed^{\LE{}}\Theta_L$ and $[\ed^{\LE{}}_X,\ed^{\LE{}}]=0$. For the last one,
\begin{eqnarray*}
\ed^{\LE{}}_XE_L  & = & \ed^{\LE{}}_X\ed^{\LE{}}_\Delta L-\ed^{\LE{}}_XL\\
& = & \ed^{\LE{}}_\Delta \ed^{\LE{}}_XL-\ed^{\LE{}}_XL+\ed^{\LE{}}_{[X,\Delta]}L\\
& = & E_{\ed^{\LE{}}_XL} + \ed^{\LE{}}_{[X,\Delta]_{\LE{}}}L.
\end{eqnarray*}
\end{proof}
We can see that not every section is admissible in any of the three
senses due to the presence of additional terms, so one must impose
restrictions on them. An important class of admissible
sections is given by complete lifts, as the next result shows.
\begin{prop}\label{completeadmissible}
For any $X\in \Gamma(\tau^E)$ its complete lift $X^c$ is an admissible section for $\w_L,\Theta_L$ and $E_L$.
\end{prop}
\begin{proof}
Since complete and vertical lifts generate $\Gamma(\tau^{\LE{}})$, in the hypothesis of Lemma \ref{lemachanges} we
can make $Y=a^c+b^v$ with
$a,b\in\Gamma(\tau^E)$. Then, using the standard properties
of lifts of sections and of the vertical endomorphism,$$
[X^c,Y]_{\LE{}}  =  [X^c,a^c+b^v]_{\LE{}}=([X,a]_E)^c+([X,b]_E)^v,$$ and
therefore $S([X^c,Y]_{\LE{}})=([X,a]_{\LE{}})^v$. Now, since
$S(Y)=a^v$ we have that
$[X^c,S(Y)]_{\LE{}}=[X^c,a^v]_{\LE{}}=([X,a]_E)^v$, and
hence
$$S([X^c,Y]_{\LE{}})=[X^c,S(Y)]_{\LE{}},$$
which proves that $X^c$ is an admissible section for $\w_L$ and $\theta_L$. Finally, that $X^c$ is an
admissible section for $E_L$ follows from $[X^c,\Delta]_{\LE{}}=0$ (see Lemma \ref{liouvillecomplete}). \end{proof}
In view of Proposition \ref{completeadmissible}  we hereafter restrict ourselves to those admissible sections
which are complete lifts of sections of $\tau^E$. We are specially interested in sections which are symmetries of the dynamics, i.e. such that $[X^c,Z_L]_{\LE{}}=0$.
We start by identifying a necessary condition for such a section.

\begin{prop}
Let $X\in\Gamma(\tau^E)$ satisfy $[X^c,Z_L]_{\LE{}}=0$. Then there is a Lagrangian $L'$ defined by $L'=\ed^{\LE{}}_{X^c}L$ such that
$$\iota_{Z_L}\w_{L'}=\ed^{\LE{}}E_{L'}.$$
If $L'$ is regular, then $L$ is equivalent to $L'$, i.e. $Z_{L}=Z_{L'}$.
\end{prop}
\begin{proof}
First notice that, for any section $Y\in\Gamma(\tau^{\LE{}})$ we have, applying $\ed^{\LE{}}_Y$ to \eqref{EL},
$$(\ed^{\LE{}}_Y\w_L)(Z_L)+\w_L(\ed^{\LE{}}_YZ_L)=\ed^{\LE{}}\ed^{\LE{}}_YE_L.$$
The previous equation reduces, when using $Y=X^c$, Proposition \ref{completeadmissible} and the fact that $X^c$ is a symmetry of $Z_L$, to
$$\iota_{Z_L}\w_{\ed^{\LE{}}_{X^c}L}=\ed^{\LE{}}E_{\ed^{\LE{}}_{X^c}L}.$$
Calling $L'=\ed^{\LE{}}_{X^c}L$ and assuming that $L'$ is regular, this implies $Z_{L'}=Z_L.$
\end{proof}
\subsection{One-Parameter Families of Gauge Equivalent Lagrangians} We now study  how imposing a condition
on a complete lift of a section (which is an admissible section by
Proposition \ref{completeadmissible}) we can generate a
one-parameter family of Lagrangians which are gauge equivalent to
a given one.
\begin{thm}\label{thm one param gauge}
Let $L\in C^\infty(E)$ and $X\in\Gamma (\tau^{E})$. Assume that
$\ed^{\LE{}}_{X^c}L=\widehat{\beta}+W\circ\tau^E$ with
$\ed^E\beta=0$ and $\ed^EW=0$. Then for each $t\in \R$ such that the  flow of
$\rho^{\LE{}}(X^c)$,
$\phi^t_{\rho^{\LE{}}(X^c)}$, is defined, the Lagrangian  $L_t=L\circ \phi^t_{\rho^{\LE{}}(X^c)}$ is gauge
equivalent to $L$.
\end{thm}
\begin{proof}
By the formula for the relationship between a vector field and its flow, we have
$$(\rho^{\LE{}}(X^c)L)\circ\phi^t_{\rho^{\LE{}}(X^c)}=\frac{d}{dt}(L\circ\phi^t_{\rho^{\LE{}}(X^c)}).$$
Using $\rho^{\LE{}}(X^c)L=\ed^{\LE{}}_{X^c}L=\widehat{\beta}+W\circ\tau^E$, $L_t=L\circ \phi^t_{X^c}$ and $\rho^{\LE{}}(X^c)=X^\text{comp}$
this is equivalent to \begin{equation}
\label{parameterequation}\frac{d}{dt}L_t =\widehat{\beta}(\phi^t_{X^\text{comp}})+W(\tau^E(\phi^t_{X^\text{comp}})),\end{equation}
with the initial condition $L_0=L$. Using the fact that the flows of $X^\text{comp}$ and $\rho^E(X)$ are $\tau^E$-related, we have that
\eqref{parameterequation} is equivalent to
$$\frac{d}{dt}L_t(a)=\langle\beta (\phi^t_{\rho^E(X)}(x)),(\phi^t_{X^\text{comp}}(a)\rangle+W(\phi^t_{\rho^E(X)}(x)),\quad\forall x\in M,\,a\in{\tau^E}^{-1}(x).$$
Note that for fixed $t$ the  flow $\phi^t_{X^\text{comp}}$ is an automorphism
of $E$ covering $\phi^t_{\rho^E(X)}$, and
thus $[\ed^E,\phi^{t*}_{X^\text{comp}}]=0$ for all $t$,  where $\phi^{t*}_{X^\text{comp}}\beta$ is defined by
$\langle\phi^{t*}_{X^\text{comp}}\beta,a\rangle=\langle\beta (\phi^t_{\rho^E(X)}(x)),(\phi^t_{X^\text{comp}}(a)\rangle$, with $\beta\in\Omega^1(E)$. Therefore
 solving \eqref{parameterequation} with the initial condition $L_0=L$ we obtain
$$L_t(a)  =  L(a)+\widehat{\overline{\beta}}(a)+\overline{W}\circ\tau^E,$$
where $\langle\overline{\beta},a\rangle=\int_0^t\langle\phi^{t'*}_{X^\text{comp}}\beta,a\rangle dt'$ and $\overline{W}(a)=\int_0^t
(\phi^{t'*}_{\rho^E(X)}W)(x)dt'$.
Now, since $\ed^E$ commutes with $\phi^{t*}_{X^\text{comp}}$ we have $\ed^E\overline{\beta}=0$ and $\ed^E\overline{W}=0$,
from the closeness of $\beta$ and $W$.
\end{proof}

\subsection{N\"other's Theorem} We now prove a N\"other theorem for Lagrangian mechanics on Lie algebroids.
If the complete lift of a section in $\Gamma(\tau^E)$ generates a one-parameter family of gauge equivalent Lagrangians with the particular data $\beta=d^Eh, W=0$ using the notation of Theorem \ref{thm one param gauge}, then we can associate to this family a conserved quantity for the dynamics.

\begin{thm}\label{noether}
Let $X\in\Gamma(\tau^E)$. If there exists $h\in C^\infty(M)$ such that
$$\ed^{\LE{}}_{X^c}L=\widehat{\ed^Eh},$$
then for any $K\in C^\infty (E)$ with $\ed^{\LE{}}K=0$, the function
$$f=\iota_{X^c}\Theta_L-h\circ\tau^E+K$$
is a conserved quantity for the dynamics associated to $L$.
\end{thm}
\begin{proof}
First, note that since $X^c$ is a complete lift, then
$$\iota_{X^c}\ed^{\LE{}}E_L  =  \ed^{\LE{}}_{X^c}E_L =E_{\ed^{\LE{}}_{X^c}L}  =E_{\widehat{d^Eh}}= \ed^{\LE{}}_\Delta(\widehat{\ed^Eh})-\widehat{\ed^Eh}=0,$$
where the last step follows from Lemma \ref{lemadelta}.

Also, we have for the Poincar\'e 1-form

$$\ed^{\LE{}}_{X^c}\Theta_L = \Theta_{\ed^{\LE{}}L}=\Theta_{\widehat{\ed^Eh}}=\ed^{\LE{}}(h\circ\tau^E),$$
where the last step follows from the properties in \eqref{props prolongation}.\\
It follows, using the definition \eqref{EL} of the section $Z_L$, that
\begin{equation}\label{thm noether 1}
\iota_{X^c}\iota_{Z_L}\w_L=\iota_{X^c}\ed^{\LE{}}E_L=\ed^{\LE{}}_{X^c}E_L=0.\end{equation}
Now, putting $f=\iota_{X^c}\Theta_L-h\circ\tau^E+K$, as in the statement, we have
\begin{equation}\label{thm noether 2}\rho^{\LE{}}(Z_L)f=\iota_{Z_L}\ed^{\LE{}}f=\iota_{Z_L}\ed^{\LE{}}\left(\iota_{X^c}\Theta_L-h\circ\tau^E +K \right).
\end{equation}
Using \eqref{magic formula} we can write
\begin{eqnarray*}\ed^{\LE{}}(h\circ\tau^E)  & = &  \Theta_{\widehat{\ed^Eh}}=\ed^{\LE{}}_{X^c}\Theta_L  =\iota_{X^c}\ed^{\LE{}}\Theta_L+\ed^{\LE{}}\iota_{X^c}\Theta_L\\ & = & -\iota_{X^c}\w_L+\ed^{\LE{}}\iota_{X^c}\Theta_L.\end{eqnarray*}
From this last expression,
$$\ed^{\LE{}}(\iota_{X^c}\Theta_L-h\circ\tau^E+K)=\iota_{X^c}\w_L.$$
Therefore,  \eqref{thm noether 2} is equivalent to
$$\rho^{\LE{}}(Z_L)f=\iota_{Z_L}\iota_{X^c}\w_L=0,$$
where we have used \eqref{thm noether 1}. This proves that $f$ is constant along the dynamics generated by $L$.
\end{proof}

\section{An example. The rigid body}
We will illustrate the applicability of Theorem \ref{noether} with an example from Classical Mechanics: the rigid body with an axis of symmetry. For that, let $E$ be the trivial Lie algebroid $\mathfrak{so}(3)$ with base a point. Then, it is immediate to compute from the definitions that\begin{itemize}
\item   $\rho^E=0$ and $\tau^E=0$,
\item  $\mathcal{L}E=\{(a,v_b)\, a\in\g,\,v_b\in T\g\}=\g\times\g\times\g$,
\item  $\rho^{\mathcal{L}E}(a,v_b)=v_b$ and
\item  $\tau^{\mathcal{L}E}(a,v_b)=b$.
\end{itemize}
Let $\{\xi_1,\xi_2,\xi_3\}$ be a basis for $\mathfrak{so}(3)$. In its associated linear coordinates, an element $y\in\mathfrak{so}(3)$ is expressed as $y=(y^1,y^2,y^3)$. We will consider the purely kinetic Lagrangian funcion
$$L(y)=\frac12\left(I_1(y^1)^2+I_2(y^2)^2+I_3(y^3)^2\right),$$
corresponding to the quadratic form $\mathbb{I}$ on $\mathfrak{so}(3)$, given in this basis by the expression
$$\mathbb{I}=\operatorname{diag}(I_1,I_2,I_3)$$

The equations of motion for this Lagrangian system are easily obtained from \eqref{localEL}, making $\rho^i_\alpha=0$ and $C_{\alpha,\beta}^\gamma=\epsilon_{\alpha,\beta,\gamma}$, the structure constants for $\mathfrak{so}(3)$. They are

$$\begin{array}{ccc}
I_1\dot y^1 & = & (I_2-I_3)y^2y^3\\
I_2\dot y^2 & = & (I_3-I_1)y^1y^3\\
I_3\dot y^3 & = & (I_1-I_2)y^1y^2,
\end{array}$$
where $I_1,I_2,I_3$ are constants.

These are, of course, Euler's equations for a rigid body in $\mathbb{R}^3$ having $\mathbb{I}$ as its inertia tensor. Suppose now that the body has an axis of symmetry. For instance, this happens if $I_2=I_3$. It follows from the previous equations that  then $I_1\dot y^1=0$, which implies that $I_1y^1$ is conserved along the dynamics, since $I_1$ is constant. We will now show how to obtain this conservation law from Theorem \ref{noether}.

Since the base $M$ of $E$ is trivial, the basis element $\xi_1\in\mathfrak{so}(3)$ can be regarded as a section $\xi_1\in\Gamma (E)$. Then we have
  $$\xi_1^c(y)=(\xi_1,y^3\frac{\partial}{\partial y^2}-y^2\frac{\partial}{\partial y^3})\quad\text{and}\quad \xi_1^v(y)=(0,\frac{\partial}{\partial y^1}).$$

In order to check if $\xi_1$ generates a symmetry of the Lagrangian, leading to a N\"other conserved quantity we compute
$$\rho^{\mathcal{L}E}(\xi_1^c)L=(y^3\frac{\partial}{\partial y^2}-y^2\frac{\partial}{\partial y^3})L=y^2y^3(I_2-I_3)=0,$$
which is of the form $\rho^{\mathcal{L}E}(\xi_1^c)L=\widehat{\ed^Eh}$ with $h=0$.
Then from Theorem \ref{noether} the function $f=\iota_{\xi_1^C}\Theta_L-h\circ\tau^E=\iota_{\xi_1^C}\Theta_L$ is a conserved quantity. The explicit form of $f$ is
\begin{eqnarray*}f & = & (\ed^{\mathcal{L}E}L\circ S)(\xi_1^C) =  \ed^{\mathcal{L}E}L(\xi_1^V)\\ & = & \rho^{\mathcal{L}E}(\xi_1^V)L
 =  \frac{\partial}{\partial y^1}L\\
& = & I_1y^1\end{eqnarray*}

Therefore $I_1y^1$ is a N\"other conserved quantity.

\section{Non-N\"other Conserved Quantities}

In a more general setup, we may have a Lagrangian $L$ such that its
dynamics $Z_L$ is locally Hamiltonian for a different 2-form
$\w'\in\Omega^2(\tau^{\LE{}})$, which may be degenerate. That is,
$\ed^{\LE{}}_{Z_L}\w'=0$. In this case N\"other's theorem is not
applicable since there is  no  one-parameter family of gauge
equivalent Lagrangians. However in this situation we can still
obtain a family of conserved quantities by taking advantage of the fact
that $Z_L$ is a symmetry of both $\w_L$ and $\w'$.
\begin{thm}\label{haz}
Let $L$ be a regular Lagrangian and $\w'\in\Omega^2(\tau^{\LE{}})$
such that
$$\ed^{\LE{}}_{Z_L}\w'=0.$$ Then the coefficients of the
polynomial $f(\lambda)$ in one real variable defined by
\begin{equation}\label{bihamiltonian}(\w'-\lambda\,\w_L)^{\wedge\,\mathrm{rank}E}=f(\lambda)\,\w_L^{\wedge\,\mathrm{rank}E}\end{equation}
are conserved quantities for $Z_L$.
\end{thm}
\begin{proof}
Since $(\w'-\lambda\,\w_L)\in\Omega^2(\tau^{\LE{}})$, then its
 power of degree equal to $\mathrm{rank}\,E$  belongs to the top cohomology class of
$\Omega^\bullet(\tau^{\LE{}})$ and hence it must be proportional to
the orientation form $\w_L^{\wedge\,\mathrm{rank}E}$ (since $L$ is
regular) by a polynomial in $\lambda$ with coefficients in $C^\infty(E)$ which
may have zeros. Therefore \eqref{bihamiltonian} is well-defined.
Now, using the two compatibility conditions $\ed^{\LE{}}_{Z_L}\w'=0$
and $\ed^{\LE{}}_{Z_L}\w_L=0$ and the usual properties of the Lie
derivative with respect to the wedge product, we have
\begin{eqnarray*}
0 & = &
\ed^{\LE{}}_{Z_L}(\w'-\lambda\,\w_L)^{\wedge\,\mathrm{rank}E}=(\ed^{\LE{}}_{Z_L}f(\lambda))\w_L^{\wedge\,\mathrm{rank}E}
+f(\lambda)\ed^{\LE{}}_{Z_L}\w_L^{\wedge\,\mathrm{rank}E}\\
& = &
(\ed^{\LE{}}_{Z_L}f(\lambda))\w_L^{\wedge\,\mathrm{rank}E},\end{eqnarray*}
from where it follows by the non-degeneracy of $L$ that
$\ed^{\LE{}}_{Z_L}f(\lambda)=\rho^{\mathcal{L}E}(Z_L)f(\lambda)=0$. \end{proof} We  will now focus  in
the particular case when there is a function $L'$ 
equivalent to $L$ in the sense of Definition  \ref{def equivalence}, since in this case $\w_{L'}$ satisfies the
hypotheses of Theorem \ref{haz}. Note that by the definition
\eqref{omega} of $\w_L$, we have $\w_{L'}-\lambda\,\w_{L}=\w_{L'-\lambda
L}$. If $p=\mathrm{rank}\,E$, we have from
\eqref{localsymplecticform} the local expression
$$\w_L^{\wedge\,p}=p!\det \left(\frac{\partial^2 L}{\partial y^\alpha\partial y^\beta}\right)\widetilde{T}^1\wedge\ldots
\wedge\widetilde{T}^p\wedge\widetilde{V}^1\wedge\ldots\wedge\widetilde{V}^p.$$
Consequently for \eqref{bihamiltonian} we can write
$$(\w_{L'}-\lambda\,\w_L)^{\wedge\,p}=p!\det \left(\frac{\partial^2
L'}{\partial y^\alpha\partial y^\beta}-\lambda \frac{\partial^2
L}{\partial y^\alpha\partial
y^\beta}\right)\widetilde{T}^1\wedge\ldots
\wedge\widetilde{T}^p\wedge\widetilde{V}^1\wedge\ldots\wedge\widetilde{V}^p,$$
and so, in this trivialization $$f(\lambda)=\frac{\det (A'-\lambda
A)}{\det A}=\det (A'A^{-1}-\lambda I),$$ where
\begin{equation}\label{hazmatrices}A'=\left[\frac{\partial^2
L'}{\partial y^\alpha\partial y^\beta} \right]_{\alpha,\beta\in
(1,\ldots,\mathrm{rank}\,E)}\quad A=\left[\frac{\partial^2
L}{\partial y^\alpha\partial y^\beta} \right]_{\alpha,\beta\in
(1,\ldots,\mathrm{rank}\,E)} \end{equation} are locally defined
matrices. We have then proved the following result.

\begin{prop}
Let $L$ and $L'$ be dynamically equivalent Lagrangians. Then, the
coefficients of the characteristic polynomial of $A'A^{-1}$, with
$A$ and $A'$ as in \eqref{hazmatrices}, are locally defined
conserved quantities for the dynamics associated
to $L$.
\end{prop}

This result is an extension to the framework of Lagrangian  mechanics on
Lie algebroids of the results by Hojman and Harleston \cite{HH}
as explained in \cite{CI}. In fact, we can make use of the Le Verrier method of determining the characteristic
equation of a matrix (see e.g. \cite{W}): the coefficients of the characteristic
equation of a matrix $M$ are determined by the traces of the increasing powers $M^k$ by
means of Newton's equations, therefore the mentioned  traces are also
constants of the motion.

\section*{Acknowledgements}
The research of J.F.C. was partially supported by research projects MTM2006-10531 and E24/1 (DGA). The research of M.R.-O. was partially supported by the research project MTM2006-03322 and a European Marie Curie Fellowship (IEF). M.R.-O. also wishes to thank the Department of Theoretical Physics of the University of Zaragoza for their hospitality during a research visit in February 2007. Finally, we would like to thank the anonymous referees for their suggestions on the paper.

\end{document}